\documentclass[conference]{IEEEtran}
\IEEEoverridecommandlockouts

\ifCLASSINFOpdf
 
\else
 
\fi
\usepackage{comment}
\usepackage{cite}
\usepackage{amsmath,amssymb,amsfonts}
\usepackage{bm}
\usepackage{overpic}
\usepackage{algorithm}
\usepackage{algpseudocode}%
\usepackage{graphicx}
\usepackage{textcomp}
\usepackage{xcolor}
\usepackage{float}
\usepackage{amsthm}
\usepackage{graphicx}
\usepackage{epstopdf}
\usepackage{amsmath,bm}
\usepackage{amsfonts}
\usepackage{amssymb}
\usepackage{color}
\usepackage{multirow}
\usepackage{multicol}
\usepackage{soul,xcolor}
\usepackage{setspace}
\newtheorem{lemma}{Lemma}

\DeclareMathOperator{\tr}{tr}

\newcommand{\mathacr}[1]{\mathsf{#1}}

\newcommand{\vect}[1]{\mathbf{#1}}

\newcommand{\condSum}[3]{\overset{#3}{\underset{\underset{#2}{#1}}{\sum}}}

\def\diag{\mathrm{diag}}

\def\tr{\mathrm{tr}}

\def\Htran{\mbox{\tiny $\mathrm{H}$}}
\def\Ttran{\mbox{\tiny $\mathrm{T}$}}
\def\imagunit{\mathsf{j}} % Imaginary number

\begin{document}
\makeatletter
\newcommand*{\rom}[1]{\expandafter\@slowromancap\romannumeral #1@}
\makeatother

\title{User-Centric Cell-Free Massive MIMO With RIS-Integrated Antenna Arrays 
\thanks{ The work by \"O. T. Demir was supported by 2232-B International Fellowship for Early Stage Researchers Programme funded by the Scientific and Technological Research Council of T\"urkiye. E.~Bj\"ornson was supported by the  FFL18-0277 grant from the Swedish Foundation for Strategic Research.}
}

\author{\IEEEauthorblockN{\"Ozlem Tu\u{g}fe Demir\IEEEauthorrefmark{1} and
		Emil Bj\"ornson\IEEEauthorrefmark{2}}
	\IEEEauthorblockA{\IEEEauthorrefmark{1}Department of Electrical-Electronics Engineering, TOBB University of Economics and Technology, Ankara, T\"urkiye }
	\IEEEauthorblockA{\IEEEauthorrefmark{2}Department of Computer Science, KTH Royal Institute of Technology, Stockholm, Sweden  } 
 \IEEEauthorblockA{Email: ozlemtugfedemir@etu.edu.tr, emilbjo@kth.se
		}}

\maketitle

\setstretch{1.05}

\begin{abstract}
Cell-free massive MIMO (multiple-input multiple-output) is a promising network architecture for beyond 5G systems, which can particularly offer more uniform data rates across the coverage area. 
Recent works have shown how reconfigurable intelligent surfaces (RISs) can be used as relays in cell-free massive MIMO networks to improve data rates further. 
In this paper, we analyze an alternative architecture where an RIS is integrated into the antenna array at each access point and acts as an intelligent transmitting surface to expand the aperture area. This approach alleviates the multiplicative fading effect that normally makes RIS-aided systems inefficient and offers a cost-effective alternative to building large antenna arrays. We use a small number of antennas and a larger number of controllable RIS elements to match the performance of an antenna array whose size matches that of the RIS.
In this paper, we explore this innovative transceiver architecture in the uplink of a cell-free massive MIMO system for the first time, demonstrating its potential benefits through analytic and numerical contributions. The simulation results validate the effectiveness of our proposed phase-shift configuration and highlight scenarios where the proposed architecture significantly enhances data rates. 
\end{abstract}
\begin{IEEEkeywords}
	Cell-free massive MIMO, RIS-integrated antenna array, dynamic cooperation clustering, RIS configuration.%
\end{IEEEkeywords}

\section{Introduction}
Cell-free massive MIMO (multiple-input multiple-output) has established itself as a promising infrastructure for beyond 5G networks, particularly for addressing cell-edge data rate issues and providing almost uniformly data rates \cite{Ngo2018a,demir2021foundations,Buzzi2017a}. Typically, the number of antennas per access point (AP) in cell-free massive MIMO is significantly smaller than in co-located massive MIMO systems, aiming for low-cost APs that can be densely deployed. One method to enhance performance in these systems involves the use of low-cost passive reconfigurable intelligent surfaces (RISs) \cite{cell-free-RIS1, cell-free-RIS2,cell-free-RIS3}, which are generally positioned between the APs and user equipments (UEs) to reflect signals—the predominant mode of RIS deployment.

Alternatively, the RISs can serve as intelligent reflecting/transmitting surfaces incorporated into active antenna transceivers \cite{9310290}. This concept is explored in the recent studies \cite{HuangHybrid,RIS-massive-MIMO} and feature radio propagation within the transceiver. To our knowledge, this architectural concept has not yet been considered in cell-free massive MIMO systems. Note that it is different from dynamic metasurface antennas \cite{Shlezinger2024a}, which is a metamaterial-based implementation of hybrid beamforming. 

This paper examines the uplink centralized operation in cell-free massive MIMO, maintaining the typical small antenna count at each AP. However, each AP is augmented with an RIS operating in a full transmission mode during uplink. We conceptualize a system where the energy collected by each RIS is confined in an ``open-box'' structure with full reflectors bridging the active antenna array and the RIS. This configuration ensures minimal energy losses, and the combination of a large number of RIS elements with properly tuned phase-shifts facilitates data rate improvements in a cost-efficient manner using relatively few radio frequency (RF) chains. Furthermore, this innovative architecture does not require additional channel estimation, as the link between the active antenna array and the RIS is deterministic and known.

The contributions of this paper are outlined as follows:
\begin{itemize}
\item We analyze, for the first time, an RIS-integrated AP receiver in user-centric cell-free massive MIMO operation and derive the achievable uplink spectral efficiency (SE).
\item We introduce a long-term phase-shift configuration scheme designed to maximize the received signal strength at each AP. The effectiveness of this method is demonstrated through simulations, showing significantly higher data rates compared to random phase-shift selections.
\end{itemize}

\section{System Model}

In this paper, we explore a cell-free massive MIMO network comprising $L$ APs and $K$ single-antenna UEs distributed across a large coverage area. Each AP is equipped with $M$ antennas and fronted by a RIS, which includes $N$ reconfigurable elements as illustrated in Fig.~\ref{figure_RIS_integrated_AP}. This configuration is designed to harvest more uplink signal energy than conventional setups where only the $M$ active antennas are used. To this end, we postulate that $N \gg M$. It is assumed that the channel between the $M$ antennas of each AP $l$ and each UE $k$ is blocked by the front RIS. We focus on the uplink operation, where each RIS operates solely in transmission mode, thereby not reflecting any incoming energy. Consequently, the uplink signals reach the RIS and are transmitted through it towards the AP antennas, with each RIS element imparting a phase-shift during the process. We further assume that there is no energy loss in the signal transmission through the RIS; thus, the total received power is determined by the RIS area, not the antenna array area. This can be facilitated by encasing the region between the AP antennas and the RISs with full reflectors on the sides, thereby minimizing energy leakage. This assumption establishes a theoretical benchmark for system performance analysis.

\begin{figure} 
	\centering 
	\begin{overpic}[width=\columnwidth,tics=10]{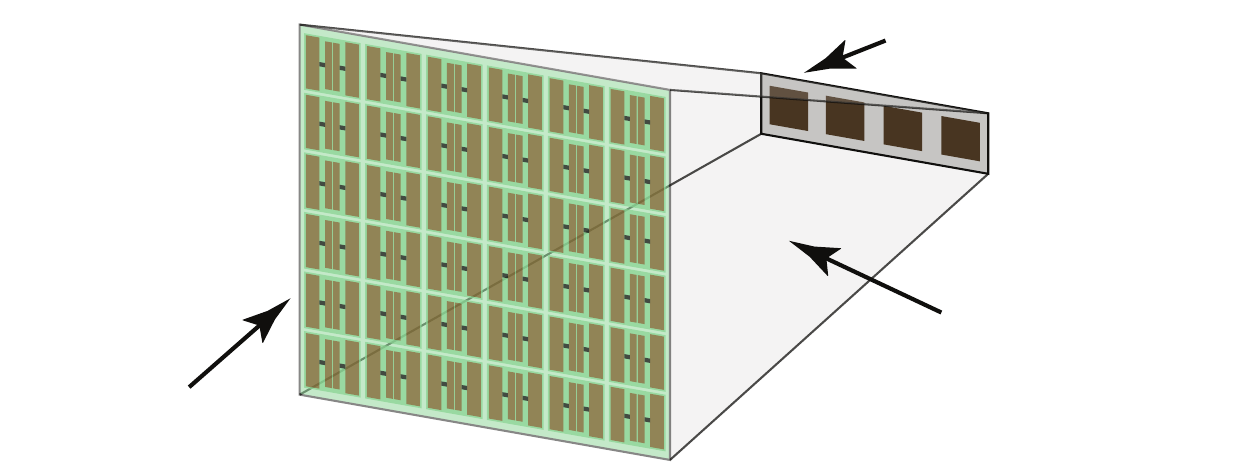}
		\put(0,2.5){\small Transmissive RIS,}
		\put(0,-1.5){\small $N$ elements}
		\put(72,34){\small $M$ active antennas}
		\put(77,10){\small Enclosed box}
	\end{overpic}  
	\caption{Illustration of an RIS-integrated antenna array.}
	\label{figure_RIS_integrated_AP} \vspace{-4mm}
\end{figure}
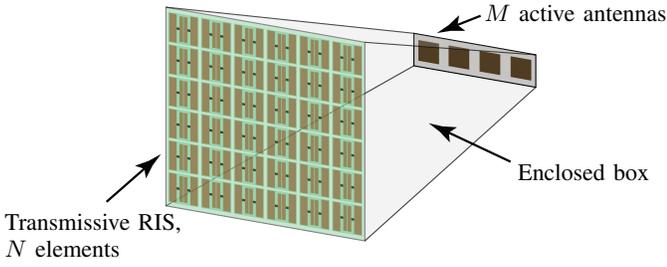

The channel between AP $l$ and its corresponding front RIS is assumed to be a Ricean fading channel consisting of a line-of-sight (LOS) radiative near-field channel plus the spatially correlated non-line-of-sight (NLOS) part corresponding to the reflections from the sides of the reflectors from the RIS to the AP. The channel matrix is denoted by $\vect{H}_l$. The power of the NLOS component is adjusted such that the norm of each column in $\vect{H}_l$ is equal to one. This adjustment ensures that the energy collected by each RIS element is distributed across all AP antennas. $\vect{H}_l$ is a constant channel due to non-varying propagation conditions between each AP and its front RIS. The LOS part of this channel can be expressed as 
\begin{align}
   \begin{bmatrix}
   \frac{\lambda}{4\pi d_{l,1,1}} e^{-\imagunit \frac{2\pi d_{l,1,1}}{\lambda}} & \ldots & \frac{\lambda}{4\pi d_{l,1,N}} e^{-\imagunit \frac{2\pi d_{l,1,N}} {\lambda}} \\
    \vdots & \ddots & \vdots \\
   \frac{\lambda}{4\pi d_{l,M,1}} e^{-\imagunit \frac{2\pi d_{l,M,1}}{\lambda}} & \ldots & \frac{\lambda}{4\pi d_{l,M,N}}e^{-\imagunit \frac{2\pi d_{l,M,N}} {\lambda}} 
    \end{bmatrix},
\end{align}
where $d_{l,m,n}$ is the propagation distance between the $m$-th antenna of AP $l$ and the $n$-th element of the RIS in front of AP $l$, and $\lambda$ is the wavelength. The AP antennas are assumed to be isotropic, but it is straightforward to consider other antennas.

The communication link between UE $k$ and RIS $l$ is described by the channel vector $\vect{h}_{kl} \in \mathbb{C}^{N}$. In this context, we apply the well-established block fading model \cite{massivemimobook}, where $\vect{h}_{kl}$ is fixed for time-frequency blocks spanning $\tau_c$ channel uses and experiences independent changes in each coherence block, following correlated Rayleigh fading: $\vect{h}_{kl} \sim \mathcal{N}_{\mathbb{C}}(\vect{0}_N,\vect{R}_{kl})$. The matrix $\vect{R}_{kl}$ characterizes the spatial correlation of the channel between RIS $l$ and UE $k$, and also accounts for the effects of large-scale fading. We assume that these correlation matrices are accessible wherever required because they are determined by the channels' long-term statistics and remain constant during transmission.

Within each coherence block, there are $\tau_p$ channel uses allocated for uplink pilot signals and $\tau_c-\tau_p$ channel uses for transmitting uplink payload data. We adopt a centralized cell-free massive MIMO operation mode \cite{demir2021foundations}, where the received signals are directly forwarded to the central processing unit (CPU). At the CPU, these signals undergo linear processing for the purpose of joint signal detection.

\subsection{Dynamic Cooperation Clustering}
In this study, we consider a user-centric cell-free network by employing the dynamic cooperation clustering (DCC) framework, initially introduced in \cite{Bjornson2013d}. Following the notation from \cite{demir2021foundations}, we use $\mathcal{M}_k \subset \{1, \ldots, L\}$ to represent the set of APs that serve UE $k$. These sets are determined by the CPU for a specific configuration and remain constant during communication. We define the DCC matrices based on the sets $\{\mathcal{M}_k; k=1,\ldots,K\}$ as 
\begin{equation}\label{eq:DCC}
\vect{D}_{kl}= \begin{cases}
\vect{I}_M, & \textrm{if } l \in \mathcal{M}_k, \\ 
\vect{0}_{M \times M}, & \textrm{if } l \not \in \mathcal{M}_k. \end{cases}
\end{equation}
Additionally, we specify $\mathcal{D}_l$ as the set containing the indices of UEs served by AP $l$:
\begin{equation}
\mathcal{D}_l = \bigg\{ k : \tr(\vect{D}_{kl})\geq 1, k \in \{ 1, \ldots, K \} \bigg\}. \label{eq:D-l}
\end{equation}

Given that the RIS is a frequency-flat device, it is essential to maintain a consistent configuration throughout each coherence block during data transmission. Furthermore, this configuration should be uniformly applied across all channel uses throughout the frequency spectrum. To achieve this, we propose adopting a long-term RIS phase-shift configuration, strategically selected to maximize signal strength based on long-term channel statistics. The proposed design for the RIS configuration will be outlined in the subsequent section, and numerical results will demonstrate its superiority over a random selection of phase-shifts.

\subsection{Channel Estimation}

In each coherence block, the channels from UE $k$ to AP $l$ for $l\in \mathcal{M}_k$ must be estimated. In large cell-free networks, it is common for the number of orthogonal pilot sequences, $\tau_p$, to be less than the number of UEs, $K$, due to limited pilot resources. Therefore, UEs are allocated pilots from a set of $\tau_p$ mutually orthogonal sequences, resulting in some UEs sharing the same pilot sequence. We let $\bm{\varphi}_k \in \mathbb{C}^{\tau_p}$ represent the pilot sequence used by UE $k$, with a norm square of $\Vert\bm{\varphi}_{k}\Vert^2=\tau_p$, and let $\mathcal{P}_k$ be the set of UEs that use the same pilot sequence as UE $k$ (including UE $k$ itself). The phase-shift configuration of RIS $l$ during pilot transmission is represented by the diagonal matrix $\vect{\Psi}_l = \diag(e^{-\imagunit\psi_{l,1}},\ldots, e^{-\imagunit\psi_{l,N}})$. Consequently, the pilot signal received at AP $l$, denoted as ${\bf Z}_l\in \mathbb{C}^{M\times \tau_p}$, is
\begin{equation}\label{eq:pilot}
 {\bf Z}_{l}=\sum_{k=1}^{K}\sqrt{\rho_p}\vect{H}_l\vect{\Psi}_l{\bf h}_{kl}\bm{\varphi}_{k}^{\Ttran}+{\bf N}_l, \vspace{-2mm}
\end{equation}
where $\rho_p$ is the pilot transmit power and the elements of noise matrix ${\bf N}_{l}\in \mathbb{C}^{M \times \tau_p}$ are independent and identically distributed (i.i.d.) $\mathcal{N}_{\mathbb{C}}(0,\sigma^2)$ random variables and $\sigma^2$ is the receiver noise variance. After correlating $ {\bf Z}_{l}$ with the pilot sequence $\bm{\varphi}_{k}$, a sufficient statistics for estimating the channel of UE $k$ is obtained as
\begin{equation}\label{eq:suff-stats} 
 {\bf z}_{kl}=\frac{{\bf Z}_l\bm{\varphi}_k^{*}}{\sqrt{\tau_p}}=\sqrt{\tau_p\rho_p}\sum_{i \in \mathcal{P}_k}\vect{H}_l\vect{\Psi}_l{\bf h}_{il}+{\bf n}_{kl},
\end{equation}
where ${\bf n}_{kl}={\bf N}_l\bm{\varphi}_k^{*}/\sqrt{\tau_p}$ has i.i.d. $\mathcal{N}_{\mathbb{C}}(0,\sigma^2)$ elements.
The minimum mean-squared error (MMSE) estimate  of ${\bf h}_{kl}$ is
\begin{align}\label{eq:lmmse} 
{\bf \widehat{h}}_{kl}&=\sqrt{\tau_p\rho_p}{\bf R}_{kl}\vect{\Psi}_l^{\Htran}\vect{H}_l^{\Htran}\nonumber\\
&\quad \cdot\left(\tau_p\rho_p\sum_{i \in \mathcal{P}_k}\vect{H}_l\vect{\Psi}_l{\bf R}_{il}\vect{\Psi}_l^{\Htran}\vect{H}_l^{\Htran}+\sigma^2{\bf I}_M\right)^{-1}{\bf z}_{kl}.
\end{align}
Using the independence of the channel estimation error and channel estimate, the error correlation matrix of the channel ${\vect{h}}_{kl}$ can be obtained as
\begin{align}
 &   \vect{C}_{kl} = \vect{R}_{kl}-\tau_p\rho_p \vect{R}_{kl}\vect{\Psi}_l^{\Htran}\vect{H}_l^{\Htran}  \nonumber \\
&\cdot\left(\tau_p\rho_p\sum_{i \in \mathcal{P}_k}\vect{H}_l\vect{\Psi}_l{\bf R}_{il}\vect{\Psi}_l^{\Htran}\vect{H}_l^{\Htran}+\sigma^2{\bf I}_M\right)^{-1}\vect{H}_l\vect{\Psi}_l\vect{R}_{kl}.
\end{align}

\vspace*{-6mm}
\begin{algorithm}[h!]
	\caption{Constrained power iteration to maximize \eqref{eq:cost}.} \label{alg:power-method}
	\begin{algorithmic}[1]
		\State {\bf Initialization:} Select $\boldsymbol{\psi}_{l}^{(0)} = [1,\ldots,1]^{\Ttran}$ and the number of iterations $I$
		\For{$i=0,\ldots,I-1$} 
		\State $\vect{w}^{(i+1)} \gets \frac{\vect{A}_l\boldsymbol{\psi}_l^{(i)}}{\| \vect{A}_l\boldsymbol{\psi}_l^{(i)}\|}$ 
		\State $\boldsymbol{\psi}^{(i+1)} = \left [e^{\imagunit \arg\left([\vect{w}^{(i+1)}]_{1}\right)},\ldots,e^{\imagunit \arg\left([\vect{w}^{(i+1)}]_{N}\right)}\right]^{\Ttran}$
		\EndFor
		\State {\bf Output:} $\boldsymbol{\psi}_l^{(I)}$
	\end{algorithmic}
\end{algorithm}

Our primary focus is on maximizing the overall received signal strength, and we will demonstrate numerically that this simple strategy is very effective. Accordingly, the objective function for AP $l$ is formulated to optimize this metric as
\begin{align}
&\sum_{k\in \mathcal{D}_l}\mathbb{E}\left\{\left\Vert\sqrt{\tau_p\rho_p}\vect{H}_l\vect{\Psi}_l\vect{h}_{kl} \right\Vert^2\right\}\nonumber\\
&= \tau_p\rho_p\tr\bigg(\vect{H}_l\vect{\Psi}_l\underbrace{\sum_{k\in \mathcal{D}_l}\vect{R}_{kl}}_{\triangleq \vect{B}_l}\vect{\Psi}_l^{\Htran}\vect{H}_l^{\Htran}\bigg). \label{eq:cost}
\end{align} 
Defining $\boldsymbol{\psi}_l=\left[e^{-\imagunit\psi_{l,1}},\ldots, e^{-\imagunit\psi_{l,N}}\right]^{\Ttran}$ and denoting the eigen-decomposition of $\vect{B}_l$ as $\vect{B}_l=\sum_{n=1}^N\lambda_{l,n}\vect{u}_{l,n}\vect{u}_{l,n}^{\Htran}$ (with $\lambda_{l,n}$ and $\vect{u}_{l,n}$ being the eigenvalues and corresponding eigenvectors), we can write the objective function in \eqref{eq:cost} as
\begin{align}
\tau_p\rho_p\boldsymbol{\psi}_l^{\Htran}\underbrace{\sum_{n=1}^N \lambda_{l,n}\diag(\vect{u}_{l,n}^*)\vect{H}_l^{\Htran}\vect{H}_l\diag(\vect{u}_{l,n})}_{\triangleq \vect{A}_l}\boldsymbol{\psi}_l
\end{align}
The total signal strength maximization problem, subject to unit modulus constraints on the entries of $\boldsymbol{\psi}_l$ for RIS connected to AP $l$, can be solved using the power iteration method  \cite{bjornson2024introduction}. The steps of this iterative procedure are outlined in Algorithm~\ref{alg:power-method}.
Later, we will utilize the output of Algorithm~\ref{alg:power-method} in selecting long-term RIS phase-shift configuration.

\section{Spectral Efficiency}

 In the uplink data phase, the received signal at AP $l$ is
 \begin{equation} \label{eq:received_AP}
 {\bf r}_l=\sum_{k=1}^K\vect{H}_l\vect{\Psi}_l{\bf h}_{kl}s_k+{\bf n}_l,
 \end{equation}
 where $s_k\in \mathbb{C}$ is the zero-mean information signal of UE $k$ with $\mathbb{E}\left\{|s_k|^2\right\}=\eta_k$, and $\eta_k>0$ is the transmit power. The receiver noise is ${\bf n}_l \sim \mathcal{N}_{\mathbb{C}}\left({\bf 0}_M,\sigma^2{\bf I}_M\right)$. We let ${\bf v}_{kl} \in \mathbb{C}^{M}$ denote the effective receive combining vector for the signal of UE $k \in \mathcal{D}_l$ at AP $l$. We define $\vect{D}_{k}  = \diag( \vect{D}_{k1} , \ldots, \vect{D}_{kL} )$ is a block-diagonal matrix with $\vect{D}_{kl}$ being defined in~\eqref{eq:DCC} as $\vect{I}_M$ for APs that serve UE $k$ and zero otherwise.

The received uplink data signals at the serving APs are jointly used at the CPU to compute an estimate $\widehat{s}_{k}$ of the signal $s_k$ transmitted by UE $k$. This is achieved by summing up the inner products between the effective receive combining vectors $\vect{D}_{kl}\vect{v}_{kl}$ and $\vect{r}_{l}$ for $l=1,\ldots,L$. This yields the estimate 
\begin{align} \notag
\widehat{s}_{k} &=  \sum\limits_{l=1}^L\widehat{s}_{kl} = \sum\limits_{l=1}^L\vect{v}_{kl}^{\Htran} \vect{D}_{kl} {\bf r}_{l} = \vect{v}_{k}^{\Htran} \vect{D}_{k} {\bf r}, \label{eq:uplink-CPU-data-estimate}
\end{align}
where $\vect{v}_{k} = [\vect{v}_{k1}^{\Ttran} \, \ldots \, \vect{v}_{kL}^{\Ttran}]^{\Ttran} \in \mathbb{C}^{LM}$ is the centralized combining vector and $\vect{r} \in \mathbb{C}^{LM}$ is the collective uplink data signal given by
\begin{equation} \label{eq:received-data-central2}
\vect{r} = \begin{bmatrix} \vect{r}_{1} \\ \vdots \\ \vect{r}_{L}
	\end{bmatrix} = \sum_{i=1}^{K} \vect{H}\vect{\Psi}\vect{h}_{i} s_i + \vect{n}
\end{equation}
with $\vect{h}_i = [\vect{h}_{i1}^{\Ttran} \, \ldots \, \vect{h}_{iL}^{\Ttran}]^{\Ttran} \in \mathbb{C}^{LN}$ and $\vect{n} = [\vect{n}_{1}^{\Ttran} \, \ldots \, \vect{n}_{L}^{\Ttran}]^{\Ttran} \in \mathbb{C}^{LM}$ being the collective channel and noise vector, respectively. Similarly, $\vect{H} = \diag(\vect{H}_1,\ldots,\vect{H}_L)$ and $\vect{\Psi}= \vect\diag(\vect{\Psi}_1,
\ldots,\vect{\Psi}_L)$. The centralized combining vector can be selected based on the knowledge of partial MMSE estimates $\{\vect{D}_{k}\vect{H}\vect{\Psi}\widehat{\vect{h}}_{i} : i= 1, \ldots,K\}$. The long-term RIS phase-shift configuration can be selected to maximize the total received signal strength by implementing Algorithm~\ref{alg:power-method}. In the following lemma, we provide an achievable SE for the considered AP architecture with centralized cell-free massive MIMO operation and a long-term RIS configuration.

\begin{lemma} \label{proposition:uplink-capacity-general}
	An achievable SE of UE $k$ in the centralized cell-free massive MIMO operation with long-term RIS phase-shift configuration is
	\begin{equation} \label{eq:uplink-rate-expression-general}
	\mathacr{SE}_{k}= \frac{\tau_c-\tau_p}{\tau_c} \mathbb{E} \left\{ \log_2  \left( 1 + \mathacr{SINR}_{k}  \right) \right\} \quad \textrm{bit/s/Hz},
	\end{equation}
	where the instantaneous effective signal-to-interference-plus-noise ratio (SINR) is given by
	\begin{equation} \label{eq:uplink-instant-SINR-level4}
	\mathacr{SINR}_{k} = \frac{ \eta_{k} \left |  \vect{v}_{k}^{\Htran} \vect{D}_{k}\vect{H}\vect{\Psi}\widehat{\vect{h}}_{k} \right|^2  }{ 
		\condSum{i=1}{i \neq k}{K} \eta_{i} \left | \vect{v}_{k}^{\Htran} \vect{D}_{k}\vect{H}\vect{\Psi}\widehat{\vect{h}}_{i} \right|^2
		+ \vect{v}_{k}^{\Htran}  {\bf{Z}}_k \vect{v}_{k}  + \sigma^2 \| \vect{D}_k \vect{v}_{k} \|^2}
	\end{equation}
	with ${\bf{Z}}_k = \sum\limits_{i=1}^{K} \eta_{i}  \vect{D}_k  \vect{H}\vect{\Psi}\vect{C}_{i} \vect{\Psi}^{\Htran}\vect{H}^{\Htran}\vect{D}_k$ and the expectation is with respect to the channel estimates. The matrix $\vect{C}_{i}$ is the error correlation matrix of the collective channel ${\vect{h}}_{i}$, which is given as $\vect{C}_{i}  = \diag( \vect{C}_{i1} , \ldots, \vect{C}_{iL} )$.
\end{lemma}
\begin{proof}
	It follows similar steps as the proof of \cite[Th.~5.1]{demir2021foundations}.
\end{proof}

The SE expression in \eqref{eq:uplink-rate-expression-general} holds for any receive combining vector $\vect{v}_{k}$, but we would like to use the one that maximizes its value.
We notice that \eqref{eq:uplink-instant-SINR-level4} has the form of a generalized Rayleigh quotient:
	\begin{align} \label{eq:uplink-instant-SINR-level4-Rayleigh}
&	\mathacr{SINR}_{k} = \nonumber\\
&\frac{ \eta_{k} \left |  \vect{v}_{k}^{\Htran} \vect{D}_{k}\vect{H}\vect{\Psi}\widehat{\vect{h}}_{k} \right|^2  }{\vect{v}_{k}^{\Htran} \left(   \condSum{i=1}{i \neq k}{K} \eta_{i}  \vect{D}_{k}\vect{H}\vect{\Psi}\widehat{\vect{h}}_{i} \widehat{\vect{h}}_{i}^{\Htran}\vect{\Psi}^{\Htran}\vect{H}^{\Htran}\vect{D}_{k} + \vect{Z}_{k}  + \sigma^2 \vect{D}_k \right) \vect{v}_{k}  
	},
	\end{align}
where we have replaced $\vect{D}_k^2$ by $\vect{D}_k$ in the last term of the denominator by using the relation $\vect{D}_k^2=\vect{D}_k$. The instantaneous SINR in \eqref{eq:uplink-instant-SINR-level4} for UE~$k$ is maximized by the MMSE combining vector
\begin{align} \label{eq:MMSE-combining}
\vect{v}_{k}^{{\rm MMSE}} &=  
\eta_{k} \Bigg( \sum\limits_{i=1}^{K} \eta_{i}  \vect{D}_{k} \vect{H}\vect{\Psi}\left( \widehat{\vect{h}}_{i} \widehat{\vect{h}}_{i}^{\Htran} + \vect{C}_{i}\right)\vect{\Psi}^{\Htran}\vect{H}^{\Htran}\vect{D}_{k} \nonumber\\&\quad+  \sigma^2 \vect{I}_{LM} \Bigg)^{\!-1}  \vect{D}_{k}\vect{H}\vect{\Psi}\widehat{\vect{h}}_{k}.
\end{align}

The results outlined above can be derived using the standard cell-free massive MIMO SE and combining formulas once the RIS configuration is fixed. The key distinction lies in the emergence of ``effective'' channels characterized by correlated Rayleigh fading, each with distinct covariance matrices.

\section{Numerical Results}\label{sec:numerical-results}

In this section, we assess the SE enhancements achieved by the proposed AP architecture with an integrated RIS, compared to conventional APs. Our analysis primarily follows the network setup detailed in \cite[Sec.~5.3]{demir2021foundations}. Unless specified otherwise, we consider a total coverage area of $1 \times 1$\,km. The APs are distributed uniformly at random within this area. Communications occur over a bandwidth of 20\,MHz, with a total receiver noise power set at $-94$\,dBm for both APs and UEs. Each UE has a maximum uplink transmit power of 100\,mW. The duration of each coherence block is $\tau_c = 200$ samples with $\tau_p=10$ samples being allocated for pilot transmission. The path losses and shadowing correlations adhere to the models described in \cite[Sec.~5.3]{demir2021foundations}. Furthermore, the joint pilot assignment and DCC algorithm described in \cite[Alg.~4.1]{demir2021foundations} is utilized to assign the pilots and define the DCC sets. The RIS at each AP is equipped with a  $6\times 6$ uniform planar array configuration with $\lambda/2$-element spacing. The P-MMSE combining vector, a scalable alternative to the MMSE combining vector, is utilized as described in \cite[Ch.~5]{demir2021foundations}. 

We will plot the cumulative distribution function (CDF) of the SE for random AP and UE locations.
We consider four different setups: i) RIS-integrated cell-free massive MIMO with $M=4$ active AP antennas in the uniform linear array configuration with $\lambda/2$ spacing. The RIS phase-shifts are optimized using the proposed Algorithm~\ref{alg:power-method}; ii) RIS-integrated cell-free massive MIMO with $M=4$ active AP antennas and random RIS phase-shifts; iii) the conventional cell-free massive MIMO (``No RIS'') with $M=4$ active AP antennas; and iv) the conventional case with $M=36$ active antennas.

\begin{figure}[t!]
	\vspace{0.1cm}
	\begin{center}
		\includegraphics[trim={0.6cm 0cm 1cm 0.6cm},clip,width=8.1cm]{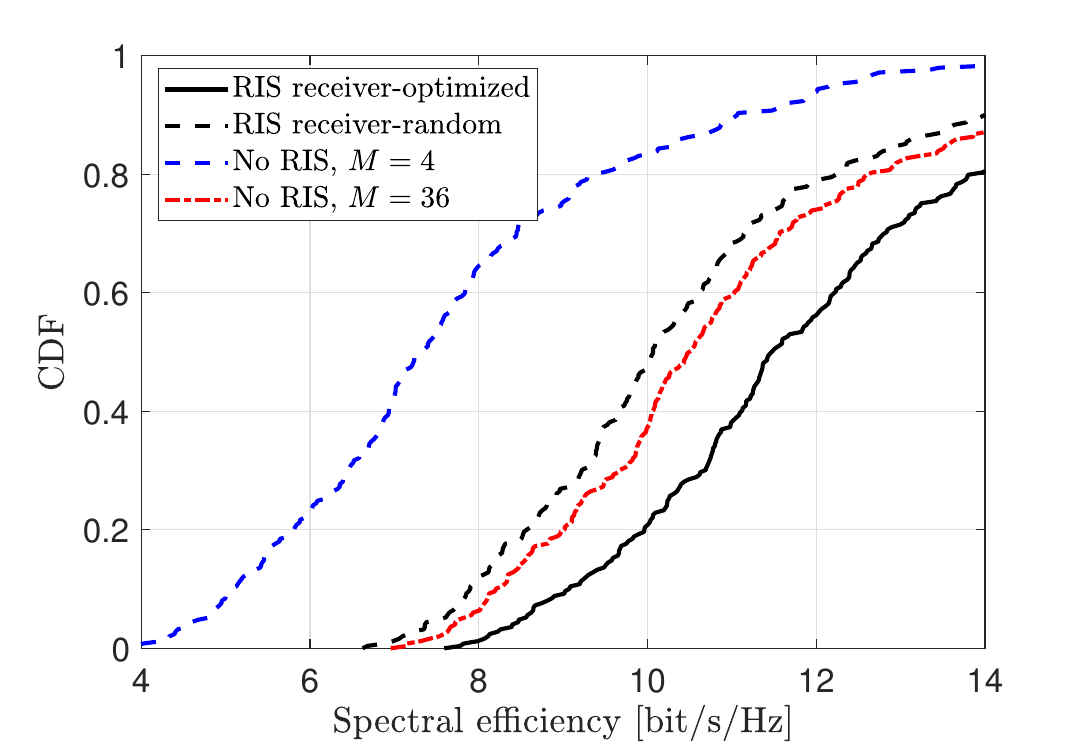}
		\vspace{-4mm}
		\caption{The CDF of the SE per UE for $L=50$ and $K=10$.} \label{fig:sim1}
	\end{center}
\vspace{-8mm}
\end{figure}

In Fig.~\ref{fig:sim1}, we begin by presenting a smaller-scale setup with $L=50$ APs and $K=10$ UEs. As illustrated, the proposed RIS-integrated cell-free architecture achieves the highest SE, surpassing even configurations with $36$ active AP antennas. This enhancement arises because, in the proposed setup with $N$ RIS elements and $M$ active array antennas, the total received signal power is proportional to $N$ and the total noise power is proportional to $M$. Therefore, when $M \ll N$, the noise power is reduced by a factor of $N/M$ (before accounting for beamforming gains). Furthermore, our proposed optimization technique yields much better SE than random phase-shifts. Most notably, compared to the configuration without RIS and with the same number of active AP antennas ($M=4$), the proposed cell-free massive MIMO setup with RIS-integrated antenna arrays achieves a 55\% improvement in median SE.

\begin{figure}[t!]
		\vspace{0.1cm}
	\begin{center}
		\includegraphics[trim={0.6cm 0cm 1cm 0.6cm},clip,width=8.1cm]{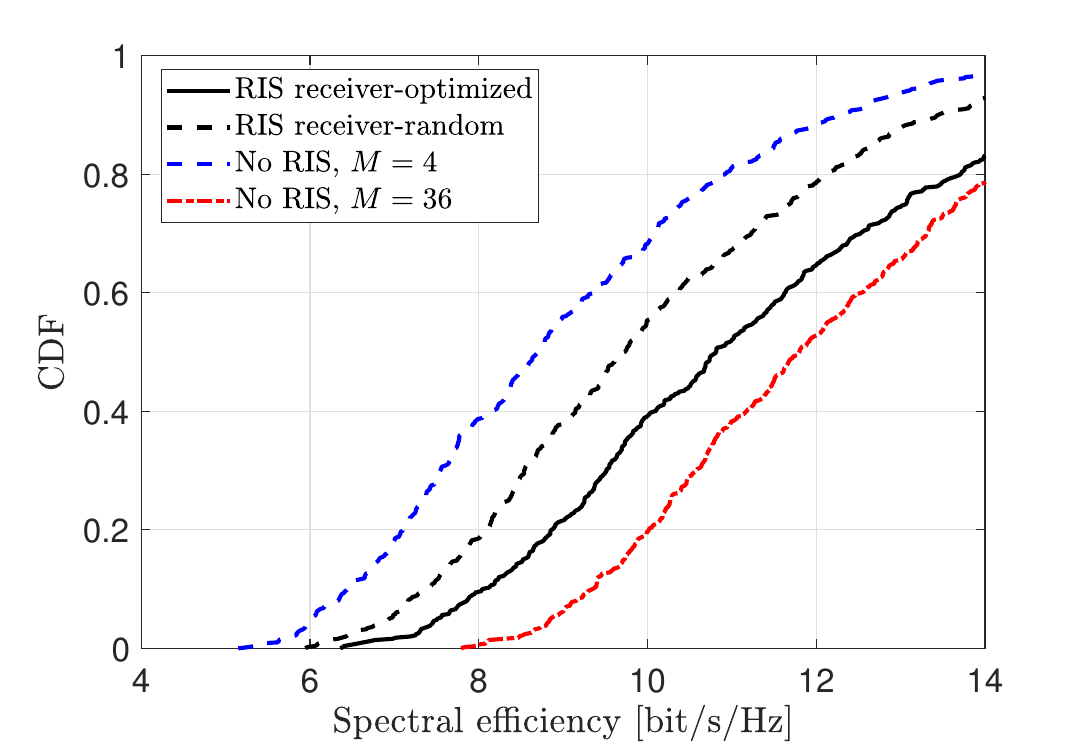}
		\vspace{-4mm}
		\caption{The CDF of the SE per UE for $L=100$ and $K=20$ in a $1 \times 1$\,km area.} \label{fig:sim2}
	\end{center}
\vspace{-8mm}
\end{figure}

In Fig.~\ref{fig:sim2}, we expand the setup to include $L=100$ APs and $K=20$ UEs within the same $1  \times 1$\,km area. A key observation is that the proposed framework's performance is somewhat diminished by the larger scale, failing to achieve the SE levels of the configuration with $M=36$ active AP antennas. Nevertheless, the proposed RIS-integrated configuration still significantly outperforms the setup that uses the same number of $M=4$ active APs without the RIS.

 We hypothesize that the diminished gains of the proposed method in the larger network are due to the more densely packed deployment compared to the initial, smaller setup of APs and UEs. To test this, we maintain a high number of APs and UEs at $L=100$ and $K=20$, respectively, but also expand the network area to $2 \times 2$\,km. Fig.~\ref{fig:sim3} displays the CDF of the SE for this configuration. The figure shows that the performance of the proposed RIS-integrated setup surpasses that of the non-RIS configuration with $M=36$ active AP antennas.

\begin{figure}[t!]
		\vspace{0.1cm}
	\begin{center}
		\includegraphics[trim={0.6cm 0cm 1cm 0.6cm},clip,width=8.1cm]{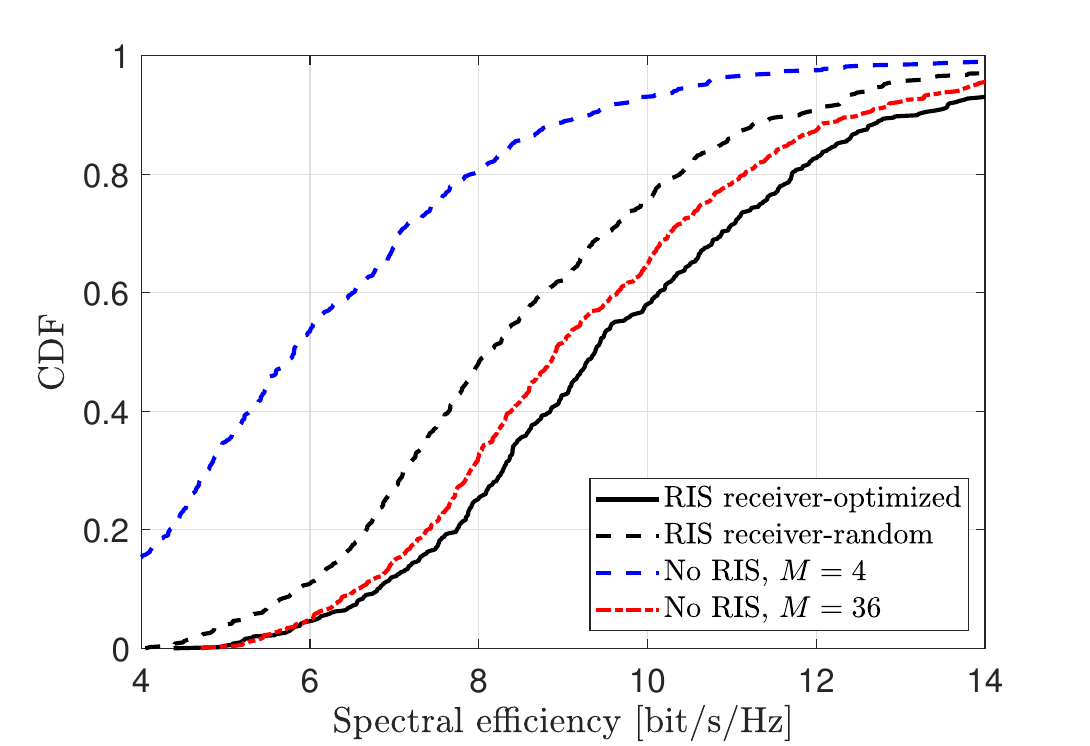}
		\vspace{-4mm}
		\caption{The CDF of the SE per UE for $L=100$ and $K=20$ in a $2 \times 2$\,km area.} \label{fig:sim3}
	\end{center}
\vspace{-8mm}
\end{figure}

\section{Conclusions}

In conclusion, this paper has introduced a novel RIS-integrated user-centric cell-free massive MIMO architecture that significantly enhances uplink SE without requiring more antennas. Our simulations have demonstrated that integrating RIS with smaller arrays of active antennas can substantially outperform traditional cell-free MIMO systems that do not contain RISs. There are even gains compared with an active array whose area matches that of the RIS. 

The analytical contribution consists of a long-term phase-shift configuration algorithm and generalization of cell-free massive MIMO methodology to the considered scenario.
In smaller-scale network setups with fewer APs and UEs, our proposed method achieved up to a 55\% improvement in median SE compared to traditional setups with the same number of active antennas but without RIS assistance. These improvements highlight the effectiveness of our phase-shift optimization algorithm, which outperforms random phase-shift configurations and supports higher SEs.

However, as network size increased, the benefits of the RIS-integrated approach were somewhat reduced, though performance remained superior to configurations without RIS technology. This suggests that while RIS-integrated antenna array technology greatly enhances performance in less dense networks, its advantage diminishes—but does not disappear—in more densely deployed networks. 
The potential reason is that the RIS provides great beamforming gains, but the spatial resolution and multiplexing capability remain limited by the number of active antennas.
Further work could explore the optimal ratio for RIS elements to active antennas in networks with different sizes and user densities.

\bibliographystyle{IEEEtran}

\bibliography{IEEEabrv,refs}

% Generated by IEEEtran.bst, version: 1.14 (2015/08/26)
\begin{thebibliography}{10}
\providecommand{\url}[1]{#1}
\csname url@samestyle\endcsname
\providecommand{\newblock}{\relax}
\providecommand{\bibinfo}[2]{#2}
\providecommand{\BIBentrySTDinterwordspacing}{\spaceskip=0pt\relax}
\providecommand{\BIBentryALTinterwordstretchfactor}{4}
\providecommand{\BIBentryALTinterwordspacing}{\spaceskip=\fontdimen2\font plus
\BIBentryALTinterwordstretchfactor\fontdimen3\font minus
  \fontdimen4\font\relax}
\providecommand{\BIBforeignlanguage}[2]{{%
\expandafter\ifx\csname l@#1\endcsname\relax
\typeout{** WARNING: IEEEtran.bst: No hyphenation pattern has been}%
\typeout{** loaded for the language `#1'. Using the pattern for}%
\typeout{** the default language instead.}%
\else
\language=\csname l@#1\endcsname
\fi
#2}}
\providecommand{\BIBdecl}{\relax}
\BIBdecl

\bibitem{Ngo2018a}
H.~Q. {Ngo}, L.~{Tran}, T.~Q. {Duong}, M.~{Matthaiou}, and E.~G. {Larsson},
  ``On the total energy efficiency of cell-free massive {MIMO},'' \emph{IEEE
  Trans. Green Commun. Net.}, vol.~2, no.~1, pp. 25--39, 2018.

\bibitem{demir2021foundations}
{\"O}.~T. Demir, E.~Bj{\"o}rnson, and L.~Sanguinetti, ``Foundations of
  user-centric cell-free massive {MIMO},'' \emph{Foundations and
  Trends{\textregistered} in Signal Processing}, vol.~14, no. 3-4, pp.
  162--472, 2021.

\bibitem{Buzzi2017a}
S.~Buzzi and C.~D'Andrea, ``Cell-free massive {MIMO}: User-centric approach,''
  \emph{{IEEE} Commun. Lett.}, vol.~6, no.~6, pp. 706--709, 2017.

\bibitem{cell-free-RIS1}
T.~Van~Chien, H.~Q. Ngo, S.~Chatzinotas, M.~Di~Renzo, and B.~Ottersten,
  ``Reconfigurable intelligent surface-assisted cell-free massive {MIMO}
  systems over spatially-correlated channels,'' \emph{IEEE Transactions on
  Wireless Communications}, vol.~21, no.~7, pp. 5106--5128, 2022.

\bibitem{cell-free-RIS2}
B.~Al-Nahhas, M.~Obeed, A.~Chaaban, and M.~J. Hossain, ``{RIS}-aided cell-free
  massive {MIMO}: Performance analysis and competitiveness,'' in \emph{2021
  IEEE International Conference on Communications Workshops (ICC Workshops)},
  2021, pp. 1--6.

\bibitem{cell-free-RIS3}
M.~Bashar, K.~Cumanan, A.~G. Burr, P.~Xiao, and M.~Di~Renzo, ``On the
  performance of reconfigurable intelligent surface-aided cell-free massive
  {MIMO} uplink,'' in \emph{GLOBECOM 2020 - 2020 IEEE Global Communications
  Conference}, 2020, pp. 1--6.

\bibitem{9310290}
V.~Jamali, A.~M. Tulino, G.~Fischer, R.~R. Müller, and R.~Schober,
  ``Intelligent surface-aided transmitter architectures for millimeter-wave
  ultra massive {MIMO} systems,'' \emph{IEEE Open Journal of the Communications
  Society}, vol.~2, pp. 144--167, 2021.

\bibitem{HuangHybrid}
A.~Huang, X.~Mu, L.~Guo, and G.~Zhu, ``Hybrid active-passive {RIS} transmitter
  enabled energy-efficient multi-user communications,'' \emph{IEEE Trans.
  Wireless Commun.}, 2024.

\bibitem{RIS-massive-MIMO}
S.~Buzzi, C.~D'Andrea, and G.~Interdonato, \emph{{RIS}-Aided Massive {MIMO}
  Antennas}.\hskip 1em plus 0.5em minus 0.4em\relax John Wiley \& Sons, Ltd,
  2023, ch.~6, pp. 117--138.

\bibitem{Shlezinger2024a}
N.~Shlezinger, G.~C. Alexandropoulos, M.~F. Imani, Y.~C. Eldar, and D.~R.
  Smith, ``Dynamic metasurface antennas for {6G} extreme massive {MIMO}
  communications,'' \emph{IEEE Wireless Communications}, vol.~28, no.~2, pp.
  106--113, 2021.

\bibitem{massivemimobook}
E.~Bj\"{o}rnson, J.~Hoydis, and L.~Sanguinetti, ``Massive {MIMO} networks:
  {Spectral}, energy, and hardware efficiency,'' \emph{Foundations and
  Trends{\textregistered} in Signal Processing}, vol.~11, no. 3-4, pp.
  154--655, 2017.

\bibitem{Bjornson2013d}
E.~Bj{\"{o}}rnson and E.~Jorswieck, ``Optimal resource allocation in
  coordinated multi-cell systems,'' \emph{Foundations and
  Trends{\textregistered} in Communications and Information Theory}, vol.~9,
  no. 2-3, pp. 113--381, 2013.

\bibitem{bjornson2024introduction}
E.~Bj{\"o}rnson and {\"O}.~T. Demir, \emph{Introduction to Multiple Antenna
  Communications and Reconfigurable Surfaces}.\hskip 1em plus 0.5em minus
  0.4em\relax Now Publishers, Inc., 2024.

\end{thebibliography}

\end{document}